\documentclass[conference]{IEEEtran}
\IEEEoverridecommandlockouts

\usepackage{bm,mathtools,amsmath,amsfonts,amssymb,amsthm,bbm,float,euscript,setspace,tikz,steinmetz,subcaption,cite,algorithmic,extarrows,enumitem}

\usepackage{mleftright}\mleftright

\usepackage[mode=buildnew]{standalone}

\newtheorem{theorem}{Theorem}
\newtheorem{lemma}{Lemma}

\newtheorem{proposition}{Proposition}

\definecolor{darkgreen}{rgb}{0, 0.5, 0}
\definecolor{darkred}{RGB}{128, 0, 0}

\newcommand{\bc}{\mathbf}

\begin{document}

\title{Zero-Error Sum Modulo Two \\with a Common Observation}
\author{Milad Sefidgaran and Aslan Tchamkerten\\Telecom Paris, Institut Polytechnique de Paris\\ \{milad.sefidgaran,aslan.tchamkerten\}@telecom-paris.fr}
\maketitle
\begin{abstract}
	This paper investigates the classical modulo two sum problem in source coding, but with a common observation:  a transmitter observes $(X,Z)$, the other transmitter observes $(Y,Z)$, and the receiver wants to compute $X \oplus Y$ without error. Through a  coupling argument, this paper establishes a new lower bound on the sum-rate when $X-Z-Y
	$ forms a Markov chain. 
\end{abstract} 

\section{introduction}
 The problem of computing the modulo two sum of binary $X$ and $Y$ observed at different transmitters was introduced by K\"orner and Marton in 1979 \cite{KornMar79}. Under the  vanishing error probability criterion, they showed the optimality of a class of linear codes when the sources have symmetric probability distributions. This implied that the optimal sum-rate for computing the sum modulo two function can be less than the optimal sum-rate for the lossless recovery of the sources as characterized by Slepian and Wolf in 1973 \cite{SlepWol73}. Later, K\"orner showed that for general sources the rate region corresponds to Slepian-Wolf's whenever $H(X\oplus Y) \geq \min (H(X),H(Y))$ (see \cite[Exercise~16.23]{CsisKor82} as this result refers to an unpublished reference).
 
 In 1983, Ahlswede and Han  \cite{AhlsHan83} used a combination of the schemes of Slepian-Wolf and K\"orner-Marton to show that it was possible to improve over the convex hull of these schemes. More recent work \cite{SefidgarnGohari15} suggests that the Ahlswede-Han scheme cannot achieve a sum-rate lower than the minimum of the sum-rates by Slepian-Wolf and K\"orner-Marton.
 
Till recently, there was no better bound than the cut-set bound $H(Y|X)+H(X|Y)$. In \cite{SefidgarnGohari15}, it was shown that when $H(X\oplus Y) < H(X|Y)+H(Y|X)$, the cut-set bound is not tight except for the cases of independent sources or sources with symmetric distribution. Recently, Nair and Wang \cite{NairWan20} established a new lower bound on the weighted sum-rate of the transmitters  which implied the optimality of the Slepian-Wolf scheme under some previously unknown conditions. Moreover, they provided sufficient conditions under which linear codes are ``weighted-sum optimal.''

Finding a better lower bound than the cut-set bound  amounts to quantify the potential penalty due to distributed processing. We note here that for the case of arithmetic sum under the zero-error performance criterion there exist better bounds than the cut-set bound as reported in \cite{MattaOster05,OrdShay15,TripRam16,TripRam18}.

\begin{figure}
\centering
\begin{tikzpicture}[scale=1]
\draw [->][thick](0.2,0.1)--(1.8,.9);
\draw [->][thick](0.2,1.9)--(1.8,1.1);
\draw [black,fill=black] (0,0) circle [radius=0.05];
\draw [black,fill=black] (0,2) circle [radius=0.05];
\draw [black,fill=black] (2,1) circle [radius=0.05];
\node [below left] at (0,0) {$(Y,Z)$};
\node [above left] at (0,2) {$(X,Z)$};
\node [right] at (2.2,1) {$X\oplus Y$};
\node [above right] at (.8,1.5) {$R_1$};
\node [below right] at (.8,.5) {$R_2$};
\end{tikzpicture}
\caption{Sum modulo two with a common observation.}
\label{fig:sumModTwoCommon}
\end{figure}
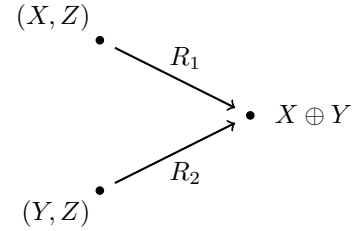
\begin{figure*}
    \begin{center}
     \includestandalone[scale=0.7]{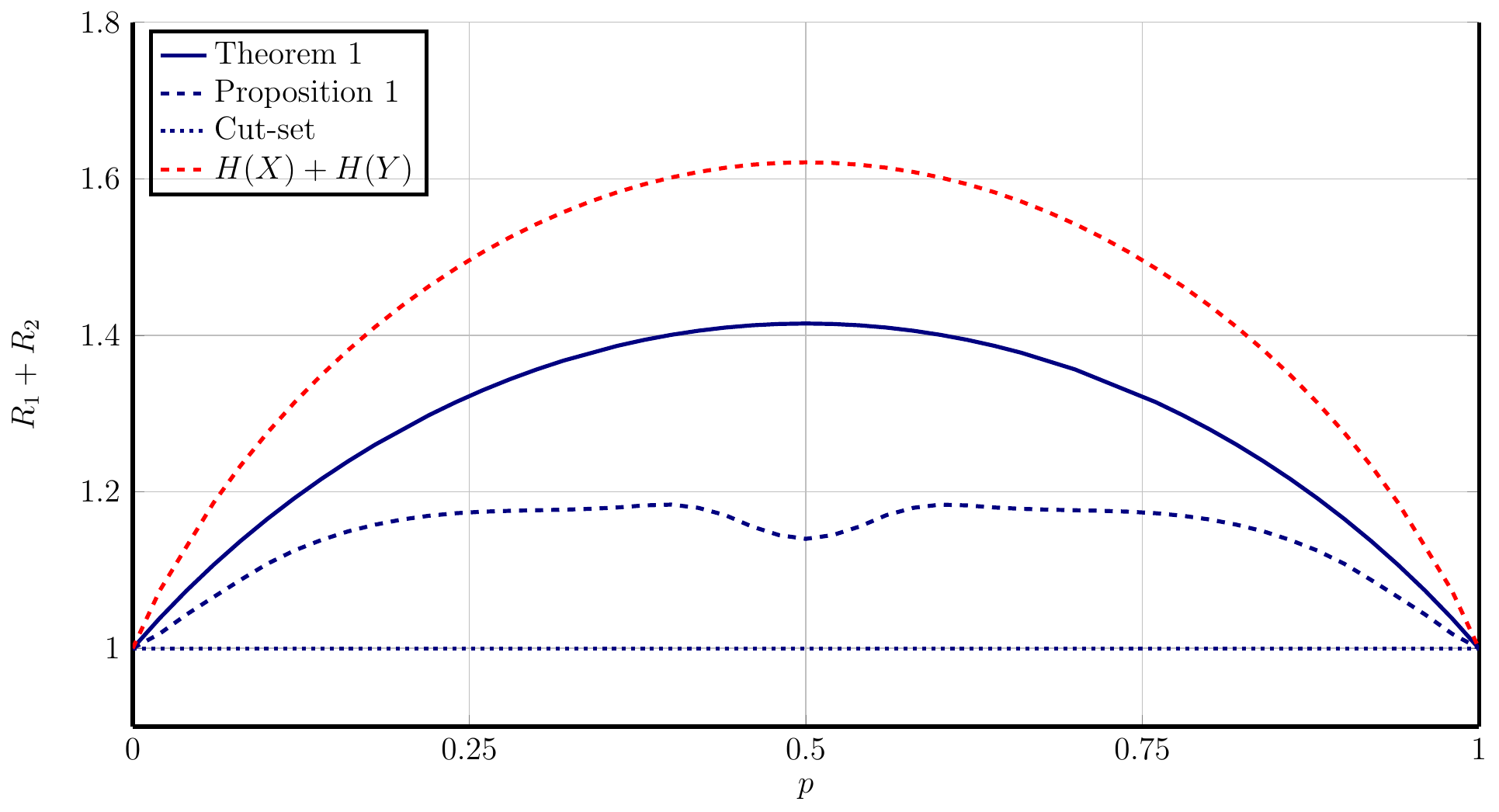}
        \end{center} 
        \caption{Comparison of different bounds for probability distribution
    \eqref{eq:probDist}.}
    \label{fig:exampleExtremeCase} 
\end{figure*}

In this paper, we consider the following variant of the sum modulo two problem for which we improve upon the cut-set bound. The transmitters have access to binary sources $(X,Z)$ and $(Y,Z)$, respectively, and use variable-length coding to send their information to the receiver who tries to compute $X \oplus Y$ without error. The sources are assumed to satisfy the Markov chain $X-Z-Y$.  We establish a lower bound on the optimal sum-rate of this problem, which improves over the cut-set bound and over an extension of a bound by Nair and Wang \cite{NairWan20} for this setup.

The rest of the paper is organized as follows.  Section~\ref{sec:preliminaries} contains preliminaries and the problem formulation. Section~\ref{sec:Result} presents the main result for a particular probability distribution of the sources. This result is proved in Section~\ref{sec:proofs}. Finally, Section ~\ref{sec:general} states the main result for arbitrary binary sources that satisfy $X-Z-Y$, and provides a proof sketch.

\section{Preliminaries} \label{sec:preliminaries}
Let $({\mathbf{x,z,y}})$ be $n$ realizations of finite support random variables $X-Z-Y$ supposed to form a Markov chain.  We consider the setup depicted in Fig.~ \ref{fig:sumModTwoCommon}.  Transmitter $1$ assigns an index $m_1({\mathbf{x,z}})\in \{1,
\ldots, L_1\}$ to its observation $({\mathbf{x,z}})$ and transmits this index to the receiver using a uniquely decodable variable length code $\mathcal{C}_1 \subseteq \{0,1\}^*$. Transmitter $2$ proceeds similarly, and assigns an index $m_2({\mathbf{y,z}})\in \{1,
\ldots, L_2\}$ to its observation $({\mathbf{y,z}})$ and transmits this index to the receiver using a uniquely decodable variable length code $\mathcal{C}_2 \subseteq \{0,1\}^*$.
Upon receiving $(m_1,m_2)$, the receiver attempts to recover the component-wise modulo two sum $\bc{x}\oplus \bc{y}$ using a decoder $\mathcal{D}(m_1,m_2) $.
The encoding procedures (indices assignment and codebooks) and the decoder are referred to as a scheme~$\mathcal{S}$. 

The sum-rate $R_1+R_2$ is said to be achievable if, for any  $\varepsilon>0$ and all $n$ sufficiently large, there exists a scheme $\mathcal{S}$  such that
\begin{align*}
    P\left( \mathcal{D}\left(M_1,M_2 \right) \neq \bc{X}\oplus \bc{Y} \right)=0
\end{align*}
and such that 
\begin{align*}
   n (R_1+R_2+\varepsilon) \geq \mathbb{E}\left[ l(c_1(M_1))+l(c_2(M_2))\right],
\end{align*}
where $l(c_i(M_i))$ denotes the length of the $M_i$-th codeword of transmitter $i$. 

A first lower bound is the standard cut-set bound
\begin{align}
    R_1+R_2 \geq& \max \left( H(X|Z)+H(Y|Z),H(X\oplus Y) \right).  \label{eq:cutSet}
\end{align}
A second lower bound is obtained through a straightforward extension of \cite[Theorem~4]{NairWan20}, originally derived under the  vanishing probability of error criterion:
\begin{proposition} \label{th:nairExtension}The sum-rate is lower-bounded as
	\begin{IEEEeqnarray*} {rll}
		R_1+R_2	\geq&  H(X|&Z)+  H(Y|Z) + H(Z)+ \\
		&\max \bigg(&\min \limits_{U-(X,Z)-Y} H(X \oplus Y|U)-H(Y,Z|U),\\
		&&\min \limits_{V-(Y,Z)-X} H(X \oplus Y|V)-H(X,Z|V)\bigg).
	\end{IEEEeqnarray*}
\end{proposition}

\section{Main Result} \label{sec:Result}
The main result, Theorem~\ref{th:mainResultZero} below, provides a new lower bound on the sum-rate. For ease of exposition this result is first stated and proved for the particular distribution
\begin{IEEEeqnarray}{rcl}
	P_{X|Z}(x=0|z=0)&=&P_{Y|Z}(y=0|z=1)=\frac{1}{2},\nonumber \\
	P_{X|Z}(x=0|z=1)&=&P_{Y|Z}(y=0|z=0)=0. \label{eq:probDist}
\end{IEEEeqnarray}
The sum-rate lower bound for the arbitrary binary sources that satisfy $X-Z-Y$ is stated in Section~\ref{sec:general}.

\begin{theorem} \label{th:mainResultZero}
Suppose $(X,Y,Z)$  satisfies \eqref{eq:probDist}. Then,
\begin{align*}
    R_1+R_2& \geq   1+L^*(p),
\end{align*}
where
\begin{align*}
L^*(p) &\triangleq\min \limits_{d \leq 2p\bar{p}} L(p,d),\\
L(p,d) &\triangleq h_b(p)+\left[d-p\cdot h_b \left(\frac{d}{2p}\right)-\bar{p}\cdot h_b \left(\frac{d}{2\bar{p}}\right)\right],
    \end{align*}
where $p\triangleq P(Z=0)$, where $\bar{p}\triangleq 1-p$, and where $h_b(p)$ denotes the binary entropy $-p\log p-(1-p)\log (1-p)$.
\end{theorem}
As we can see in Fig.~\ref{fig:exampleExtremeCase}, the lower bound given by Theorem~\ref{th:mainResultZero} improves upon both the cut-set bound and the bound given by Proposition~\ref{th:nairExtension} for all (non-trivial) values of~$p$---the shape of this latter bound around $p=1/2$ is somewhat unexpected.

\section{Proofs} \label{sec:proofs}
Throughout the proofs $\varepsilon$ is an arbitrary constant in $(0,1/2)$ that can be taken arbitrarily small for $n$ sufficiently large.
\begin{proof}[Proof of Theorem~\ref{th:mainResultZero}]
We have
\begin{align}
    R_1+R_2 &\stackrel{(a)}{\geq} \frac{1}{n}H(M_1)+\frac{1}{n}H(M_2)-\varepsilon\nonumber \\
    &\geq \frac{1}{n} H(M_1,M_2)-\varepsilon \nonumber \\
    &= \frac{1}{n} H(M_1,M_2|\bc{Z})+\frac{1}{n} I(M_1,M_2;\bc{Z})-\varepsilon \nonumber \\
    &\stackrel{(b)}{=}
    H(X|Z)+H(Y|Z)+\frac{1}{n} I(M_1,M_2;\bc{Z})-\varepsilon\nonumber \\
    &=
    H(X,Y,Z)-\frac{1}{n} H(\bc{Z}|M_1,M_2)-\varepsilon.
    \label{eq:commonITBounds}
\end{align}
Inequality $(a)$ follows from the unique decodability of the codes which implies that $\mathbb{E}\left[ l(c_i(M_i))\right]\geq H(M_i)$ \cite[Theorem~4.1]{CsisKor82}. Equality $(b)$ follows from the following lemma whose proof is deferred to the end of this section.

\begin{lemma} \label{lem:conditionalUniqueIndices} For any zero-error scheme and for sources that satisfy the Markov chain $X-Z-Y$ we have
\begin{align*}
    H(M_1,M_2|Z)=nH(X|Z)+nH(Y|Z).
\end{align*}
\end{lemma}

 The main part of the proof consists in the derivation of a good upper bound on $H(\bc{Z}|M_1,M_2)$ by introducing the following coupling of $(\bc{X},\bc{Y},\bc{Z})$.\footnote{A trivial upper bound is $H(\bc{Z})$, which results in the first term of the cut-set bound \eqref{eq:cutSet}.} Given a coding scheme, let  $(\tilde{\bc{X}},\tilde{\bc{Y}},\tilde{\bc{Z}})$ be an independent copy of $({\bc{X}},{\bc{Y}},{\bc{Z}})$ conditioned on $(M_1,M_2)$.  Hence, we have the Markov chain 
\begin{align}
    (\tilde{\bc{X}},\tilde{\bc{Y}},\tilde{\bc{Z}})-(M_1,M_2)-(\bc{X},\bc{Y},\bc{Z}) \label{eq:MarkovChain}
\end{align}
and the marginal probabilities $P_{\tilde{\bc{X}},\tilde{\bc{Y}},\tilde{\bc{Z}}}$ and $P_{\bc{X},\bc{Y},\bc{Z}}$ coincide. 

Associated with the above coupling is the $\bc{Z}$-distance of the underlying scheme which we  define as \begin{align*}
    &d_{avg}(\bc{Z})\triangleq \mathbb{E}\left[d_H(\bc{Z},\tilde{\bc{Z}})\right]\\
&={\sum \limits_{\substack{(m_1,m_2)\\(\bc{z}_1,\bc{z}_2)}} } d_H(\bc{z}_1,\bc{z}_2) p(\bc{z}_1|m_1,m_2)p(\bc{z}_2|m_1,m_2)p(m_1,m_2),
\end{align*} 
 where $d_H(\cdot,\cdot)$ denotes the normalized Hamming distance.
Now, 
\begin{align}
H(&\bc{Z}\big| M_1,M_2)\nonumber \\&\overset{(a)}{=}H(\tilde{\bc{Z}}\big| M_1,M_2,\bc{Z})\nonumber\\
&=H(M_1,M_2\big| \bc{Z},\tilde{\bc{Z}}){+}H(\tilde{\bc{Z}}|\bc{Z}) {-}H(M_1,M_2|\bc{Z})\nonumber\\
&\overset{(b)}{=}H(M_1,M_2\big| \bc{Z},\tilde{\bc{Z}}){+}H(\tilde{\bc{Z}}|\bc{Z}) {-}nH(X,Y|Z), \label{eq:zeroErrorEq1}
\end{align}
where $(a)$ holds because of \eqref{eq:MarkovChain} and where $(b)$ follows from Lemma~\ref{lem:conditionalUniqueIndices}.

The rest of the proof is divided into four parts: $i.$ upper bound on $H(M_1,M_2\big| \bc{Z},\tilde{\bc{Z}})$, $ii.$ upper bound on $H(\tilde{\bc{Z}}|\bc{Z})$, $iii.$ upper bound $H(\bc{Z}\big| M_1,M_2)$ using $i.$ and $ii.$, and finally lower bound the sum-rate in part $iv.$

\paragraph*{\sc{$i$}}  For  $H(M_1,M_2|\bc{Z},\tilde{\bc{Z}})$ we have
\begin{align} 
 H(M_1&,M_2 |\bc{Z},\tilde{\bc{Z}}) \nonumber\\
  \stackrel{(a)}{=}& H\left (M_1,M_2 |\bc{Z},\tilde{\bc{Z}},\bc{X}\oplus \tilde{\bc{X}}=\bc{Y}\oplus \tilde{\bc{Y}}=g(\bc{Z},\tilde{\bc{Z}})\right)\nonumber\\
  \leq& H\left (\bc{X},\bc{Y},M_1,M_2 |\bc{Z},\tilde{\bc{Z}},\bc{X}\oplus \tilde{\bc{X}}=\bc{Y}\oplus \tilde{\bc{Y}}=g(\bc{Z},\tilde{\bc{Z}})\right)\nonumber\\
 \stackrel{(b)}{=}& H\left (\bc{X},\bc{Y} |\bc{Z},\tilde{\bc{Z}},\bc{X}\oplus \tilde{\bc{X}}=\bc{Y}\oplus \tilde{\bc{Y}}=g(\bc{Z},\tilde{\bc{Z}})\right)\nonumber\\
 \leq &H\left(\bc{X}|\bc{Z},\tilde{\bc{Z}},\bc{X}\oplus \tilde{\bc{X}}=g(\bc{Z},\tilde{\bc{Z}})\right)\nonumber \\
 &+H\left(\bc{Y}|\bc{Z},\tilde{\bc{Z}},\bc{Y}\oplus \tilde{\bc{Y}}=g(\bc{Z},\tilde{\bc{Z}})\right). \label{eq:zeroErrorEq2}
\end{align}
Equality $(a)$ follows from Lemma~\ref{lem:SameFZeroError} below (the proof is deferred to the end of this section) and Equality $(b)$ holds since $(M_1,M_2)$ is a function of $(\bc{X},\bc{Y},\bc{Z})$.
\begin{lemma}\label{lem:SameFZeroError} For any zero-error scheme and for sources that satisfy the Markov chain $X-Z-Y$, there exists a function $g:\mathcal{Z}^n \times \mathcal{Z}^n \rightarrow \{0,1\}^n$ such that 
     \begin{align*}
 P\left(\bc{X}\oplus \tilde{\bc{X}}=g(\bc{Z},\tilde{\bc{Z}})= \bc{Y}\oplus \tilde{\bc{Y}}\right)=1.
\end{align*} 
\end{lemma}We now bound the two terms on the right-hand side of \eqref{eq:zeroErrorEq2} for the particular distribution \eqref{eq:probDist}.
Conditioned on $\bc{Z}=\bc{z}$ and $\tilde{\bc{Z}}=\tilde{\bc{z}}$, for any index $i$ such that $OR(z_i,\tilde{z}_i)= 1$, either $X_i=1$ or $\tilde{X}_i=1$ (or both). In each of these cases  $H(X_i|\bc{z},\tilde{\bc{z}},\bc{X}\oplus \tilde{\bc{X}}=g(\bc{z},\tilde{\bc{z}}))=0$, and consequently 
\begin{align*}
     H\left(\bc{X}|\bc{z},\tilde{\bc{z}},\bc{X}\oplus \tilde{\bc{X}}=g(\bc{z},\tilde{\bc{z}})\right) \leq n\hat{p}_{\bc{z},\tilde{\bc{z}}}(0,0),
\end{align*}
where $\hat{p}_{\bc{z},\tilde{\bc{z}}}(\cdot,\cdot)$ denotes the empirical distribution of $(\bc{z},\tilde{\bc{z}})$. 
Using a similar argument for $\bc{Y}$, we deduce from  \eqref{eq:zeroErrorEq2} that
\begin{align} 
 H\left (M_1,M_2 |\bc{z},\tilde{\bc{z}}\right) &\leq n(\hat{p}_{\bc{z},\tilde{\bc{z}}}(0,0)+\hat{p}_{\bc{z},\tilde{\bc{z}}}(1,1)) \nonumber\\
 &= n(1-d_H(\bc{z},\tilde{\bc{z}})), \nonumber
\end{align}
and hence,
\begin{align} 
 H\left (M_1,M_2 |\bc{Z},\tilde{\bc{Z}}\right) &\leq  n-n\sum \limits_d d
 \cdot P\left(d_H(\bc{Z},\tilde{\bc{Z}})=d\right). \label{eq:zeroErrorEq3} 
\end{align}
\paragraph*{ii} The term $ H(\tilde{\bc{Z}}|\bc{Z})$ in \eqref{eq:zeroErrorEq1} is bounded as 
\begin{align}
    H(\tilde{\bc{Z}}|\bc{Z}) \leq& H\left(\tilde{\bc{Z}},d_H(\bc{Z},\tilde{\bc{Z}})|\bc{Z}\right) \nonumber\\
    =& H\left(\tilde{\bc{Z}}|\bc{Z},d_H(\bc{Z},\tilde{\bc{Z}})\right)+H\left(d_H(\bc{Z},\tilde{\bc{Z}})|\bc{Z}\right) \nonumber\\
    \stackrel{(a)}{\leq}& H\left(\tilde{\bc{Z}}|\bc{Z},d_H(\bc{Z},\tilde{\bc{Z}})\right)+n\varepsilon \nonumber \\
    =& \sum \limits_{d} H\left(\tilde{\bc{Z}}|\bc{Z},d_H(\bc{Z},\tilde{\bc{Z}})=d\right)P(d_H(\bc{Z},\tilde{\bc{Z}})=d)\nonumber\\
    &+n\varepsilon,\label{eq:zeroErrorEq4.1}
\end{align}
where  $(a)$ holds since $d_H(\bc{Z},\tilde{\bc{Z}})$ takes at most $n+1$ values. 

To bound $H\left(\tilde{\bc{Z}}|\bc{z},d_H(\bc{z},\tilde{\bc{Z}})=d)\right)$, first suppose that $\hat{p}_{\bc{z}}(0)=p$---hence $\bc{z}$ belongs to the set of strongly typical sequences $ \mathcal{T}_{\varepsilon}^n(Z)$. For any  strongly typical sequence $\tilde{\bc{z}}$ such that $d_H(\bc{z},\tilde{\bc{z}})=d$, it can be verified that both $\hat{p}_{\bc{z},\tilde{\bc{z}}}(0,1)$ and $\hat{p}_{\bc{z},\tilde{\bc{z}}}(1,0)$ are within the range
$$[(d-\varepsilon)/2,(d+\varepsilon)/2]. $$
The number of such sequences $\tilde{\bc{z}}$ is hence upper bounded by $2^{n\left(p h_b\left(\frac{d}{2p}\right) + \bar{p}  h_b\left(\frac{d}{2\bar{p}}\right)+\varepsilon\right)}$, and therefore $$H\left(\tilde{\bc{Z}}|\bc{z},d_H(\bc{z},\tilde{\bc{Z}})=d)\right)\leq p h_b\left(\frac{d}{2p}\right) + \bar{p}  h_b\left(\frac{d}{2\bar{p}}\right)+\varepsilon.$$ A similar bound holds for any strongly typical sequence $\bc{z}$ where $\hat{p}_{\bc{z}}(0)=p\pm\delta$ and $0<\delta<\epsilon$. Consequently, \eqref{eq:zeroErrorEq4.1} can be bounded as
\begin{align}
    H(\tilde{\bc{Z}}|\bc{Z}) \leq& n\sum \limits_{d} \left(p h_b\left(\frac{d}{2p}\right) + \bar{p}  h_b\left(\frac{d}{2\bar{p}}\right)\right)P\left(d_H(\bc{Z},\tilde{\bc{Z}})=d\right)\nonumber \\
    &+n\varepsilon. \label{eq:zeroErrorEq4}
\end{align}

\paragraph*{iii}Combining  \eqref{eq:zeroErrorEq1}, \eqref{eq:zeroErrorEq3}, and \eqref{eq:zeroErrorEq4} and using the fact that $H(X,Y|Z)=1$, the term $H(\bc{Z}|M_1,M_2)$ can be bounded as
\begin{align}
    \frac{1}{n} H(\bc{Z}&|M_1,M_2) -\varepsilon\nonumber\\
    \leq&\sum \limits_{d} \left(p h_b\left(\frac{d}{2p}\right) + \bar{p}  h_b\left(\frac{d}{2\bar{p}}\right)-d\right)P\left(d_H(\bc{Z},\tilde{\bc{Z}})=d\right)\nonumber \nonumber\\
    \leq & p h_b\left(\frac{d_{avg}(\bc{Z})}{2p}\right) + \bar{p}  h_b\left(\frac{d_{avg}(\bc{Z})}{2\bar{p}}\right)-d_{avg}(\bc{Z})\nonumber\\
    =& h_b(p)-L\big(p,d_{avg}(\bc{Z})\big). \label{eq:zeroErrorEqFinal}
\end{align}
where the second inequality follows from
Jensen's inequality and the concavity of the  entropy function. \paragraph*{iv} Finally, combining  \eqref{eq:commonITBounds} and \eqref{eq:zeroErrorEqFinal} yields
\begin{align*}
    R_1+R_2    \geq&  1+L\big(p,d_{avg}(\bc{Z})\big)-\varepsilon\\
    \geq &  1+\min \limits_{d\leq 2p\bar{p}} L\big(p,d\big)-\varepsilon,
\end{align*}
where the second inequality follows from:
\begin{lemma} \label{lem:boundOnHamDis} For any scheme we have $d_{avg}(\bc{Z}) \leq  2p\bar{p}$. 
\end{lemma}
This concludes the proof of Theorem~\ref{th:mainResultZero}.
\end{proof}
\begin{proof}[Proof of Lemma~\ref{lem:conditionalUniqueIndices}]
 Fix $\bc{Z}=\bc{z}$. Let $\bc{x}_1, \bc{x}_2,\bc{y}$ be such that  $p(\bc{x}_1|\bc{z})>0$, $p(\bc{x}_2|\bc{z})>0$, and $p(\bc{y}|\bc{z})>0$. Because of the Markov chain $X-Z-Y$, we have $p(\bc{x}_1,\bc{y}|\bc{z})>0$ and $p(\bc{x}_2,\bc{y}|\bc{z})>0$. Hence, if $(\bc{x}_1,z)$ and $(\bc{x}_2,\bc{z})$ were assigned the same index, the receiver would not be able to perfectly differentiate between $\bc{x}_1\oplus \bc{y}$ and $\bc{x}_2\oplus \bc{y}$. Thus, for any fixed $\bc{Z}=\bc{z}$,  different indices should be assigned to different $\bc{x}$ sequences having nonzero conditional probability $p(\bc{x}|\bc{z})$---and similarly for the $\bc{y}$ sequences. This implies that conditioned on $\bc{Z}=\bc{z}$, there is a one-to-one mapping between $(M_1,M_2)$ and $(\bc{X},\bc{Y})$, which yields 
 \begin{align*}
 H(M_1,M_2|\bc{Z})=H(\bc{X},\bc{Y}|\bc{Z})=nH(X|Z)+nH(Y|Z).    
 \end{align*}
 \end{proof}
 
\begin{proof}[Proof of Lemma~\ref{lem:SameFZeroError}]
First we show that
\begin{align}
    P\left(\bc{X}\oplus \bc{Y}=\tilde{\bc{X}}\oplus \tilde{\bc{Y}}\right)=1 \label{eq:equalF}.
\end{align}
 By contradiction, suppose there exist pairs $(\bc{x},\bc{y})$ and $(\tilde{\bc{x}},\tilde{\bc{y}})$ such that $\bc{x}\oplus \bc{y}\neq \tilde{\bc{x}}\oplus \tilde{\bc{y}}$ and $p(\bc{x},\bc{y},\tilde{\bc{x}},\tilde{\bc{y}})>0$. Due to the Markov chain $(\tilde{\bc{X}},\tilde{\bc{Y}})-(M_1,M_2)-(\bc{X},\bc{Y})$, there exists an index pair $(m_1,m_2)$ such that 
\begin{align*}     p(\bc{x},\bc{y}|m_1,m_2)  p(\tilde{\bc{x}},\tilde{\bc{y}}|m_1,m_2)>0.
\end{align*}
This contradicts the zero-error assumption.

We now argue that there exists a function $g(\bc{Z},\tilde{\bc{Z}})$ such that 
    \begin{align*}
        P\left(\bc{X}\oplus \tilde{\bc{X}}=g(\bc{Z},\tilde{\bc{Z}})\right)=1.
    \end{align*}
    By contradiction suppose that such a function does not exist. Then, there exist  $(\bc{x}_1,\bc{y}_1,\bc{z})$, $(\bc{x}_2,\bc{y}_2,\bc{z})$, $(\tilde{\bc{x}}_1,\tilde{\bc{y}}_1,\tilde{\bc{z}})$, and $(\tilde{\bc{x}}_2,\tilde{\bc{y}}_2,\tilde{\bc{z}})$ such that
    \begin{align*}
        p(\bc{x}_1,\bc{y}_1,\bc{z},\tilde{\bc{x}}_1,\tilde{\bc{y}}_1,\tilde{\bc{z}})&>0,\\
        p(\bc{x}_2,\bc{y}_2,\bc{z},\tilde{\bc{x}}_2,\tilde{\bc{y}}_2,\tilde{\bc{z}})&>0,
    \end{align*} 
    and such that $\bc{x}_1\oplus \tilde{\bc{x}}_1 \neq \bc{x}_2\oplus\tilde{\bc{x}}_2$.
    From \eqref{eq:equalF}, $$ \bc{x}_1\oplus \bc{y}_1 =  \tilde{\bc{x}}_1  \oplus \tilde{\bc{y}}_1$$ and $$\bc{x}_2\oplus \bc{y}_2=  \tilde{\bc{x}}_2 \oplus \tilde{\bc{y}}_2.$$ Moreover,  \begin{align*}
        p(\bc{x}_1,\bc{y}_2,\bc{z},\tilde{\bc{x}}_1,\tilde{\bc{y}}_2,\tilde{\bc{z}})>0,
    \end{align*}
    and hence by \eqref{eq:equalF} we have $\bc{x}_1\oplus \tilde{\bc{x}}_1 = \bc{y}_2\oplus  \tilde{\bc{y}}_2$. The  equalities 
    \begin{align*}
    \bc{x}_2\oplus \bc{y}_2=  \tilde{\bc{x}}_2 \oplus \tilde{\bc{y}}_2&,~
    \bc{x}_1\oplus \tilde{\bc{x}}_1 = \bc{y}_2\oplus  \tilde{\bc{y}}_2,
    \end{align*}
    contradicts $   \bc{x}_1\oplus \tilde{\bc{x}}_1 \neq \bc{x}_2 \oplus  \tilde{\bc{x}}_2$.
    
A similar argument shows that there exists a function $g'(\bc{Z},\tilde{\bc{Z}})$ such that
$        P\left(\bc{Y}\oplus \tilde{\bc{Y}}=g'(\bc{Z},\tilde{\bc{Z}})\right)=1.$
   From \eqref{eq:equalF} it then follows that $g'(\bc{Z},\tilde{\bc{Z}})=g(\bc{Z},\tilde{\bc{Z}})$.
\end{proof}

\begin{proof}[Proof of Lemma~\ref{lem:boundOnHamDis}] 
The Markov chain $\bc{Z}-(M_1,M_2)-\tilde{\bc{Z}}$ implies that $Z_i$ has the same conditional distribution as  $\tilde{Z}_i$ given $(m_1,m_2)$. Hence, define
\begin{align*}
    p(i|m_{1,2})\triangleq P(Z_i=0|m_1,m_2)=P(\tilde{Z}_i=0|m_1,m_2).
\end{align*}
By the Markov chain $Z_i-(M_1,M_2)-\tilde{Z}_i$,
$$\mathbb{E}\left[d_H(Z_i,\tilde{Z}_i) \big| m_1,m_2\right] =2p(i|m_{1,2})(1-p(i|m_{1,2})).$$
The $\bc{Z}$-distance can then be written as
\begin{align}
    d_{avg}(\bc{Z})&=\mathbb{E}\left[d_H(\bc{Z},\tilde{\bc{Z}}) \right] \nonumber \\
     &= \sum \limits_{m_1,m_2} p(m_1,m_2) \mathbb{E}\left[d_H(\bc{Z},\tilde{\bc{Z}}) \big| m_1,m_2\right] \nonumber \\
    &=\frac{1}{n} \sum \limits_{m_1,m_2} p(m_1,m_2) \sum \limits_{i=1}^n \mathbb{E}\left[d_H(Z_i,\tilde{Z}_i) \big| m_1,m_2\right] \nonumber \\
    &= \frac{2}{n}\sum \limits_{i=1}^n\sum \limits_{m_1,m_2} p(m_1,m_2) p(i|m_{1,2})(1-p(i|m_{1,2}))\nonumber\\
     &\stackrel{(a)}{\leq} \frac{2}{n}\sum \limits_{i=1}^n P(Z_i=0)(1-P(Z_i=0))\nonumber\\
    &= 2p\bar{p}, \nonumber
\end{align}
where $(a)$ follows from Jensen's inequality applied to the concave function $x-x^2$ and by noting that \begin{align*}
    P(Z_i=0)=\sum \limits_{m_1,m_2} p(m_1,m_2) p(i|m_{1,2}). &\qedhere
\end{align*}
\end{proof}

\section{Generalization of Theorem~\ref{th:mainResultZero}} \label{sec:general}
In this section, we state our result for general sources and provide a sketch of its proof. 

\begin{theorem} \label{th:mainResultZeroGeneral}
For any zero-error scheme and for sources that satisfy the Markov chain $X-Z-Y$, we have
\begin{align}
    R_1+R_2 \geq H(X|Z)+H(Y|Z)+L^*(p), \label{eq:generalBound}
\end{align}
where
\begin{align*}
    L^*(p) &\triangleq\min \limits_{d \leq 2p\bar{p}} L(p,d).
\end{align*}
Here, $L(p,d)$ is redefined as:
\begin{IEEEeqnarray}{rcl}
    L(p,d) {\triangleq}& 
    H(X,&Y,Z)-p  h_b\left(\frac{d'}{p}\right)-\bar{p} h_b\left(\frac{d'}{\bar{p}}\right){-}
\nonumber\\
    &{\max \limits_{\substack{\bc{u},\bc{v},w}}}\Bigg[& (p-d') h_b \left(\frac{u_1}{p-d'}\right){+}(\bar{p}-d') h_b \left(\frac{u_4}{\bar{p}-d'}\right)   \nonumber\\
    &&+2 g h_b \left(\frac{u_2}{w}\right)+2 (d'-w) h_b \left(\frac{u_3}{d'-w}\right)\nonumber\\
    &&+ (p-d') h_b \left(\frac{v_1}{p-d'}\right){+}(\bar{p}-d') h_b \left(\frac{v_4}{\bar{p}-d'}\right)   \nonumber\\
    &&+2 g h_b \left(\frac{v_2}{w}\right)+2 (d'-w) h_b \left(\frac{v_3}{d'-w}\right)\Bigg], \label{eq:LDefinition}
\end{IEEEeqnarray}
where $d'=d/2$, $\bc{u}\triangleq(u_1,u_2,u_3,u_4)$, $\bc{v}\triangleq(v_1,v_2,v_3,v_4)$, and where the maximization is taken over $0 \leq w \leq d'$ and  $\bc{u}$ and $\bc{v}$ such that
\begin{align}
    u_1 \leq p-d',&~u_2\leq w,~u_3\leq d'-w,~u_4\leq \bar{p}-d',\nonumber \\
    v_1\leq p-d',&~v_2\leq w,~v_3\leq d'-w,~v_4\leq \bar{p}-d', \label{eq:uvRange}
\end{align}
and such that
\begin{align}
    u_1+u_2+u_3&=p \cdot P(X=0|Z=0), \nonumber\\
    u_2-u_3+u_4&= \bar{p}\cdot  P(X=0|Z=1)+w-d',\nonumber\\
    v_1+v_2+v_3&=p \cdot P(Y=0|Z=0),\nonumber\\
    v_2-v_3+v_4&= \bar{p}\cdot  P(Y=0|Z=1)+w-d'. \label{eq:uvLinearC}
\end{align}
\end{theorem}
For the probability distribution given by \eqref{eq:probDist}, Definition~\eqref{eq:LDefinition} reduces to the definition of $L(p,d)$ in Theorem~\ref{th:mainResultZero}.

It can be verified that the sum-rate lower bound in  Theorem~\ref{th:mainResultZeroGeneral} is tight in the cases of independent $(X,Y,Z)$, constant $X$, and identical sources $X{=}Y{=}Z$. The optimal sum-rates are equal to  $H(X,Y)$, $H(Y)$, and $0$, respectively.

\begin{proof}[Proof sketch of  Theorem~\ref{th:mainResultZeroGeneral}]
All the steps up to  \eqref{eq:zeroErrorEq2} carry to the general case. We now show how to bound 
\begin{align}
    H\Big(\bc{X}|\bc{Z},\tilde{\bc{Z}},\bc{X}{\oplus} \tilde{\bc{X}}{=}g(\bc{Z},\tilde{\bc{Z}})\Big){+}H\left(\bc{Y}|\bc{Z},\tilde{\bc{Z}},\bc{Y}{\oplus} \tilde{\bc{Y}}{=}g(\bc{Z},\tilde{\bc{Z}})\right) \label{tobound}
\end{align}
for general sources. Pick  strongly typical sequences $\bc{z},\tilde{\bc{z}}\in \mathcal{T}_{\varepsilon}^n(Z)$ such that $d_H(\bc{z}$, $\tilde{\bc{z}})=d$ and consider the function $g(\bc{z}$, $\tilde{\bc{z}})$ given by Lemma~\ref{lem:SameFZeroError}. Note that this  Lemma implies that  $P(\tilde{\bc{X}}=\bc{X}\oplus g(\bc{z}$, $\tilde{\bc{z}})|\bc{z}$, $\tilde{\bc{z}})=1$. Define
\begin{align*}
    \mathcal{A}_X(\bc{z},\tilde{\bc{z}},g(\bc{z},\tilde{\bc{z}})) &{\triangleq} \Big\{ \bc{x}\colon  (\bc{x},\bc{z}),  (\bc{x}\oplus g(\bc{z},\tilde{\bc{z}}),\tilde{\bc{z}})\in \mathcal{T}_{\varepsilon}^n(X,Z)\Big\},\\
    \mathcal{P}_X(\bc{z},\tilde{\bc{z}},g(\bc{z},\tilde{\bc{z}}))&{\triangleq} \Big\{ \hat{p}_{\bc{x}|\bc{z},\tilde{\bc{z}},g(\bc{z},\tilde{\bc{z}})}(\cdot|\cdot)\colon \bc{x} \in \mathcal{A}_X(\bc{z},\tilde{\bc{z}},g(\bc{z},\tilde{\bc{z}})) \Big\}.
\end{align*}
For $\hat{p}(\cdot|\cdot) \in \mathcal{P}_X(\bc{z},\tilde{\bc{z}},g(\bc{z},\tilde{\bc{z}}))$, define
\begin{IEEEeqnarray*}{c}
\alpha_{\hat{p}(\cdot|\cdot)} {\triangleq} \big|\left\{\bc{x} \colon  \bc{x} {\in} \mathcal{A}_X(\bc{z},\tilde{\bc{z}},g(\bc{z},\tilde{\bc{z}})),\hat{p}_{\bc{x}|\bc{z},\tilde{\bc{z}},g(\bc{z},\tilde{\bc{z}})}(\cdot|\cdot)=\hat{p}(\cdot|\cdot) \right\}\big|.
\end{IEEEeqnarray*}

It can be proved that
\begin{IEEEeqnarray}{rl}
    H\Big(\bc{X}|\bc{Z}=\bc{z},\tilde{\bc{Z}}=&\tilde{\bc{z}},\bc{X}\oplus \tilde{\bc{X}}=g(\bc{z},\tilde{\bc{z}})\Big)\nonumber \\
    \leq& \log_2\left(\big|\mathcal{A}_X(\bc{z},\tilde{\bc{z}},g(\bc{z},\tilde{\bc{z}}))\big|\right)+n\varepsilon \nonumber \\
    \leq& \max_{g(\bc{z},\tilde{\bc{z}})} \max_{\substack{\hat{p}(\cdot|\cdot)\\\hat{p}(\cdot|\cdot) \in \mathcal{P}_X(\bc{z},\tilde{\bc{z}},g(\bc{z},\tilde{\bc{z}}))}}\log_2\left(\alpha_{\hat{p}(\cdot|\cdot)}\right)+n\varepsilon.\nonumber
\end{IEEEeqnarray}
Moreover, it can be shown that, in the limit of large $n$, the above maximizers $g^*(\bc{z},\tilde{\bc{z}})$ and $\hat{p}^*(\cdot|\cdot)$ satisfy:
\begin{align}
    \hat{p}_{\bc{z},\tilde{\bc{z}},g^*(\bc{z},\tilde{\bc{z}})}(0,0,1)&=\hat{p}_{\bc{z},\tilde{\bc{z}},g^*(\bc{z},\tilde{\bc{z}})}(1,1,1)=0,\nonumber\\
    \hat{p}_{\bc{z},\tilde{\bc{z}},g^*(\bc{z},\tilde{\bc{z}})}(0,1,0)&=\hat{p}_{\bc{z},\tilde{\bc{z}},g^*(\bc{z},\tilde{\bc{z}})}(1,0,0), \nonumber \\
    \hat{p}^*(0|0,1,0)=\hat{p}^*(0|1,0,0)&,~
    \hat{p}^*(0|0,1,1)=\hat{p}^*(1|1,0,1). \nonumber
\end{align}
A similar upper bound is obtained for $H\big(\bc{Y}|\bc{Z}{=}\bc{z},\tilde{\bc{Z}}{=}\tilde{\bc{z}},\bc{Y}{\oplus} \tilde{\bc{Y}}{=}g(\bc{z},\tilde{\bc{z}})\big)$. It then follows that \eqref{tobound} is upper bounded by the $\max \limits_{\substack{\bc{u},\bc{v},w}}[\ldots]$ term on the right-hand side of  \eqref{eq:LDefinition}, where $w \triangleq \hat{p}_{\bc{z},\tilde{\bc{z}},g(\bc{z},\tilde{\bc{z}})}(0,1,0)$, and
\begin{align*}
     u_1 \triangleq \hat{p}_{\bc{x}|\bc{z},\tilde{\bc{z}},g(\bc{z},\tilde{\bc{z}})}(0|0,0,0),&~ u_2 \triangleq \hat{p}_{\bc{x}|\bc{z},\tilde{\bc{z}},g(\bc{z},\tilde{\bc{z}})}(0|0,1,0),\\
     u_3 \triangleq \hat{p}_{\bc{x}|\bc{z},\tilde{\bc{z}},g(\bc{z},\tilde{\bc{z}})}(0|0,1,1),&~ u_4 \triangleq \hat{p}_{\bc{x}|\bc{z},\tilde{\bc{z}},g(\bc{z},\tilde{\bc{z}})}(0|1,1,0),
\end{align*}
and where $v_1, v_2,v_3,v_4$ are defined similarly but with respect to  $\hat{p}_{\bc{y}|\bc{z},\tilde{\bc{z}},g(\bc{z},\tilde{\bc{z}})}(\cdot|\cdot)$.
\end{proof}

\bibliographystyle{IEEEtran}
\bibliography{references}

\end{document}